\RequirePackage{fix-cm}
\documentclass[smallextended]{svjour3}       
\smartqed  
\usepackage{graphicx}
\usepackage{amsfonts}
\usepackage{amsmath}
\usepackage{amssymb}

\begin{document}

\title{Novel Constructions of Mutually Unbiased Tripartite Absolutely Maximally Entangled Bases}
\author{Tian Xie \and  Yajuan Zang \and Hui-Juan Zuo$^*$ \and Shao-Ming Fei}
\institute{ T. Xie \and H.-J. Zuo \at
              School of Mathematical Sciences, Hebei Normal University, Shijiazhuang, 050024, China \\
              Hebei Key Laboratory of Computational Mathematics and Applications, Shijiazhuang, 050024, China \\
              Hebei Mathematics Research Center, Shijiazhuang, 050024, China\\
              \email{huijuanzuo@163.com}
        \and
           Y. Zang \and S.-M. Fei \at
            School of Mathematical Sciences, Capital Normal University, Beijing 100048, China\\
                      \email{feishm@cnu.edu.cn}
           }

\date{Received: date / Accepted: date}

\maketitle
\begin{abstract}
We develop a new technique to construct mutually unbiased tripartite absolutely maximally entangled bases. We first explore the tripartite absolutely maximally entangled bases and mutually unbiased bases in $\mathbb{C}^{d} \otimes \mathbb{C}^{d} \otimes \mathbb{C}^{d}$ based on mutually orthogonal Latin squares. Then we generalize the approach to the case of $\mathbb{C}^{d_{1}} \otimes \mathbb{C}^{d_{2}} \otimes \mathbb{C}^{d_{1}d_{2}}$ by mutually weak orthogonal Latin squares. The concise direct constructions of mutually unbiased tripartite absolutely maximally entangled bases are remarkably presented with generality. Detailed examples in $\mathbb{C}^{3} \otimes \mathbb{C}^{3} \otimes \mathbb{C}^{3},$ $\mathbb{C}^{2} \otimes \mathbb{C}^{2} \otimes \mathbb{C}^{4}$ and $\mathbb{C}^{2} \otimes \mathbb{C}^{5} \otimes \mathbb{C}^{10}$ are provided to illustrate the advantages of our approach.

\end{abstract}

\section{Introduction}
Quantum entanglement is a striking feature of quantum mechanics. In particular, maximally entangled states play vital roles in quantum information processing such as quantum computation, quantum secure communication and quantum key distribution$^{[1-8]}$.
The mutually unbiased bases (MUBs) play an important role in quantum information and quantum computation such as quantum state tomography$^{[9,10]}$, cryptographic protocols$^{[11,12]}$ and quantification of wave-particle duality in multi-path interferometers$^{[13]}$.
Two orthonormal bases are said to be mutually unbiased if the transition probabilities from each state in one basis to all states of the other basis are the same irrespective of which pair of states is chosen. Namely, if the physical system is prepared in a state of the first basis, then all measurement outcomes are equally probable when one conducts a measurement that probes for the states of the second basis. Such unbiasedness plays an important role in quantum secure communication as the eavesdroppers cannot obtain any useful information.
In recent years, the researches on mutually unbiased bases mainly remain on the bipartite systems. Less is known about the construction of mutually unbiased absolutely maximally entangled bases in the multipartite systems.

For bipartite system $\mathbb{C}^{d}\otimes \mathbb{C}^{d^{\prime}}$ $(d \leq d^{\prime})$, a pure state $|\psi\rangle$ is said to be a maximally entangled if and only if for any given orthonormal complete basis $\{|i_{A}\rangle\}$ of subsystem A, there exists an orthonormal basis $\{|i_{B}\rangle\}$ of subsystem B such that $|\psi\rangle$ can be written as $|\psi\rangle=\frac{1}{\sqrt{d}}\sum\limits_{i=0}^{d-1}|i_{A}\rangle\otimes|i_{B}\rangle$ $^{[14]}$. In recent years, the research upon mutually unbiased maximally entangled bases (MUMEBs) has attracted more and more attention. Tao ${et~al}$. $^{[15]}$ first exhibited a method of constructing MUMEBs in $\mathbb{C}^{d} \otimes \mathbb{C}^{kd}$. Liu ${et~al}$. $^{[16]}$ presented a construction of $d-1$ MUMEBs in $\mathbb{C}^{d} \otimes \mathbb{C}^{d}$. Then Xu $^{[17]}$ put forward a construction of $2(d-1)$ MUMEBs in $\mathbb{C}^{d} \otimes \mathbb{C}^{d}$. In $^{[18]}$, Song ${et~al}$. constructed $d+1$ MUBs in $\mathbb{C}^{d} \otimes \mathbb{C}^{d}$, which include $d-1$ mutually unbiased maximally entangled bases and two mutually unbiased product bases. For multipartite system, Zhang ${et~al}$. $^{[19]}$ studied the construction of unextendible maximally entangled bases and mutually unbiased bases.

\par The concept of the Latin square approximately originated from the problems concerning the movement and disposition of pieces on a chess board. However, the earliest reference to the use of such squares concerning the arrangement of thirty-six officers of six different ranks and regiments in a square phalanx was proposed by Euler in 1779. Then, Ronald Fisher $^{[20]}$ noted that Latin squares and families of mutually orthogonal Latin squares could be used systematically for experimental design. Recently, Latin square has been utilized as a tool to solve critical problems in the fields such as mathematics and computer science$^{[21,22]}$.

\par In this paper, we concentrate on absolutely maximally entangled bases (AMEBs) and mutually unbiased bases (MUBs) in tripartite quantum systems $\mathbb{C}^{d} \otimes \mathbb{C}^{d} \otimes \mathbb{C}^{d}$ and $\mathbb{C}^{d_{1}} \otimes \mathbb{C}^{d_{2}} \otimes \mathbb{C}^{d_{1}d_{2}}$. We first introduce some basic definitions and properties in Section $2$. In Section $3$, we construct a pair of mutually unbiased absolutely maximally entangled bases (MUAMEBs) based on a pair of mutually orthogonal Latin squares (MOLS). In Section $4$, we put forward a general construction by mutually weak orthogonal Latin squares. We also present the corresponding examples in $\mathbb{C}^{3} \otimes \mathbb{C}^{3} \otimes \mathbb{C}^{3},$ $\mathbb{C}^{2} \otimes \mathbb{C}^{2} \otimes \mathbb{C}^{4}$ and $\mathbb{C}^{2} \otimes \mathbb{C}^{5} \otimes \mathbb{C}^{10}$. We come to the conclusion in the last section.

\section{Preliminaries}
We firs introduce some basic definitions and related results used in this paper.
\par {\bf Definition 1.} Two orthonormal bases $\mathcal{B}_{1}=\{{|\phi_{i}\rangle}\}^{d}_{i=1}$ and $\mathcal{B}_{2}=\{{|\psi_{i}\rangle}\}^{d}_{i=1}$ of $\mathbb{C}^{d}$ are called mutually unbiased if
\begin{equation}
\nonumber
 |\langle\phi_{i}|\psi_{j}\rangle|=\frac{1}{\sqrt{d}},~~ i,j=1,2,\ldots,d.
\end{equation}

A set of orthonormal bases ${\{\mathcal{B}_{1}, \mathcal{B}_{2},\ldots,\mathcal{B}_{m}\}}$ in $\mathbb{C}^{d}$ is called mutually unbiased bases
(MUBs) if every pair of the bases in the set is mutually unbiased.
Denote $N(d)$ the maximal number of MUBs in $\mathbb{C}^{d}$. It is proved that $N(d)\leq d +1$$^{[10]}$ and $N(d)=d +1$ if $d$ is a prime power. Whereas $d$ is a composite number, $N(d)$ is not yet completely known.

Concerning the absolutely maximally entangled states introduced by Facchi ${et~al}$$^{[23]}$, consider a bipartition $(A,\bar{A})$ of a system $S$, where $A\subset S$, $\bar{A}=S\setminus A$, $S = \{1,2,\ldots,n\}$ and $1\leq n_{A}\leq n_{\bar{A}} $, with $n_{A}=|A|$ the number of the parties in the subsystem $A$.

{\bf Definition 2$^{[19]}$.} A state $\rho$ is absolutely maximally entangled if and only if the reduced state is maximally mixed under all possible bipartitions $(A,\bar{A})$:
$$\rho_{A}=Tr_{\bar{A} }(\rho)=\frac{I}{n_{A}},$$
where $I$ is the corresponding identity matrix.

{\bf Definition 3.} A Latin square of order $n$ (or side $n$) is an $n\times n$ array in which each cell contains a single entry from an $n$-set $S$, such that each entry occurs exactly once in each
row and exactly once in each column$^{[24,25]}$.
\par A Latin square $L$ may also be represented as a set of ordered triples $L(i,j,k)$, which means that entry $k$ occurs in cell $(i,j)$ of the Latin square. We also denote it as $L(k)$ in short. A transversal in a Latin square of order $n$ is a set of $n$ cells, one from each row and column, containing each of the $n$ entries exactly once.

{\bf Definition 4.} Two Latin squares $L=(l_{i,j})$ (on entry set $S$) and $K=(k_{i,j})$ (on entry set $S'$) of the same order are mutually orthogonal if  every
element in $S\times S'$ occurs exactly once among the $n^{2}$ pairs $(l_{i,j} ,k_{i,j})$, $1\leq i,j\leq n$.
\par A set of Latin squares $L_{1},L_{2},\ldots,L_{m}$ is mutually orthogonal Latin squares (a set of $m$ MOLS), if $L_{i}$ and $L_{j}$ are orthogonal for $1\leq i<j\leq m$$^{[24,25]}$.

Remarkably, the maximal value $M(d)$ of MOLS of order $d$ is no more than $d-1$. If $d=p^{n}$, $d\geq3$, where $p$ is a prime number and $n\in \mathbb{Z}$, then $M(d)=d-1$. A set of $d-1$ MOLS of order $d$ is called a complete set of MOLS of order $d$. When $d$ is a composite number except for prime power, the existence about mutually orthogonal Latin squares of order $d$ is still uncertain.

Next we introduce the classical case of mutually weak orthogonal quantum Latin squares (see Definition 10$^{[26]}$).

{\bf Definition 5.} (Mutually weak orthogonal Latin squares) Given a pair of Latin squares $L$ and $K$ with entries $l_{ij}$ and $k_{ij}$, respectively. They are mutually weak orthogonal if there is exactly one intersection in any two rows from $L$ and $K$, i.e., for all $1\leq i,j\leq d$, there exists only one $s$ for $1\leq s\leq d$ such that $l_{is}=k_{js}$.

A set of Latin squares $L_{1},L_{2},\ldots,L_{m}$ is mutually weak orthogonal Latin squares, or a set of $m$ MWOLS, if $L_{i}$ and $L_{j}$ are mutually weak orthogonal for $1\leq i<j\leq m$.
The examples of MOLS and MWOLS of order 3 are given in Fig. 1.
$${\begin{tabular}{c c c}
{\small \begin{tabular}{|c|c|c|}
\hline
0 & 2 & 1\\
\hline
1 & 0 & 2\\
\hline
2 & 1 & 0\\
\hline
\end{tabular}} & {\small \begin{tabular}{|c|c|c|}
\hline
0 & 1 & 2\\
\hline
1 & 2 & 0\\
\hline
2 & 0 & 1\\
\hline
\end{tabular}}
\end{tabular}}$$
$$\text{Fig}. 1~A~pair~of~\text{MOLS}~and~\text{MWOLS}~of~order~3$$

\section{Construction of MUAMEBs in $\mathbb{C}^{d} \otimes \mathbb{C}^{d} \otimes \mathbb{C}^{d}$}
In this section, based on a pair of mutually orthogonal Latin squares, we construct a pair of mutually unbiased absolutely maximally entangled bases (MUAMEBs) in $\mathbb{C}^{d} \otimes \mathbb{C}^{d} \otimes \mathbb{C}^{d}$ $(d\neq2,6)$.

\begin{theorem}\label{1}
For any given Latin square of order $d$, there is an absolutely maximally entangled basis (AMEB) in $\mathbb{C}^{d}\otimes\mathbb{C}^{d}\otimes \mathbb{C}^{d}$.
\end{theorem}

\begin{proof} Suppose that $\mathcal{L}_{s}$ is a Latin square of order $d$. We can construct an absolutely maximally entangled basis in $\mathbb{C}^{d} \otimes \mathbb{C}^{d} \otimes \mathbb{C}^{d}$, which consists of the following two parts:
$$\left\{
\begin{array}{l}
|\phi_{n,m,0}\rangle=\frac{1}{d}\sum\limits_{k=0}^{d-1}\left(\omega_{d}^{nk}|k\rangle \otimes \sum\limits_{i=0}^{d-1}\omega_{d}^{mi}|i j\rangle_k\right)~~~~(*)\\ |\phi_{n,m,l}\rangle=\frac{1}{d}\sum\limits_{k=0}^{d-1}\left(\omega_{d}^{nk}|k\rangle \otimes \sum\limits_{i=0}^{d-1}\omega_{d}^{mi}|ij\rangle_{k\oplus_{d} l}\right)~~~~(**), ~l=1,2,\ldots,d-1,
\end{array}
\right.$$
where $0\leq n,m\leq d-1$, $\omega_{d}$=$e^{\frac{2\pi\sqrt{-1}}{d}}$, $|ij\rangle_k$ stands for $|i\rangle\otimes|j\rangle$ with $i$ and $j$ representing the indices of the row and the column of the entry $k$ in $(*)$, while $|ij\rangle_{k\oplus_{d} l}$ stands for $|i\rangle\otimes|j\rangle$ with $i$ and $j$ representing the indices of the row and the column of the entry $k\oplus_{d} l$ in $(**)$.

(i)~Absolutely maximal entanglement.
$$\centering \begin{array}{clll}
&\rho_{1}=tr_{23}(|\phi_{n,m,0}\rangle\langle\phi_{n,m,0}|)\\
&~~~=\frac{1}{d^2}tr_{23}\left(\sum\limits_{k,k'=0}^{d-1}\omega_{d}^{n(k-k')}|k\rangle \langle k'|\otimes \sum\limits_{i,i'=0}^{d-1}\omega_{d}^{m(i-i')}|ij\rangle_k \langle i' j'|_{k'}\right)\\
\hspace{-2.2cm}&~~~=\frac{1}{d^{2}}\sum\limits_{k,k^{\prime}=0}^{d-1}\omega_{d}^{n(k-k^{\prime})}|k\rangle\langle k^{\prime}|tr\left(\sum\limits_{i,i^{\prime}=0}^{d-1}|ij\rangle_k \langle i^{\prime}j^{\prime}|_{k'}\right)\\
\hspace{-2.2cm}&~~~=\frac{1}{d}\sum\limits_{k,k^{\prime}=0}^{d-1}\omega_{d}^{n(k-k^{\prime})}|k\rangle\langle k^{\prime}|\delta_{k,k'}\\
\hspace{-2.2cm}&~~~=\frac{I}{d}.
\end{array}
$$
For other two reduced states, we also have $\rho_{2}=\rho_{3}=\frac{I}{d}$ similarly. Thus, $\{|\phi_{n,m,0}\rangle, 0\leq n,m\leq d-1\}$ is a set of absolutely maximally entangled states. Similarly, we can prove that $\{{|\phi_{n,m,l}\rangle}$, $0\leq n,m\leq d-1$, $1\leq l\leq d-1\}$ are also absolutely maximally entangled states.

(ii)~Orthogonality.
$$\begin{array}{clll}
\hspace{1cm}\langle\phi_{n^{\prime},m^{\prime},0}|\phi_{n,m,0}\rangle&=\frac{1}{d^{2}}\sum\limits_{k,k^{\prime}=0}^{d-1}\omega_{d}^{nk-n^{\prime}k^{\prime}}\langle k^{\prime}|k\rangle\sum\limits_{i,i^{\prime}=0}^{d-1}\omega_{d}^{mi-m^{\prime}i^{\prime}}(|i^{\prime}j^{\prime}\rangle_{k'},|ij\rangle_{k})\\
\hspace{-1cm}&=\frac{1}{d^{2}}\sum\limits_{k=0}^{d-1}\omega_{d}^{(n-n^{\prime})k}\sum\limits_{i=0}^{d-1}\omega_{d}^{(m-m^{\prime})i}\\
\hspace{-1cm}&=\delta_{n,n^{\prime}}\delta_{m,m^{\prime}}.
\end{array}$$
Thus for a fixed $l$, $\{|\phi_{n,m,l}\rangle,0\leq n,m\leq d-1\}$ are mutually orthogonal.  And for any $0\leq l\neq l'\leq d-1$, it is easy to check that $\{|\phi_{n,m,l}\rangle,0\leq n,m\leq d-1\}$ and $\{|\phi_{n,m,l'}\rangle,0\leq n,m\leq d-1\}$ satisfy the orthogonal relations because the entries of the same cell are different.
\end{proof}

\begin{theorem}\label{2}
For any given pair of mutually orthogonal Latin squares of order $d$, there exists a pair of MUAMEBs in $\mathbb{C}^{d} \otimes \mathbb{C}^{d} \otimes \mathbb{C}^{d}$.
\end{theorem}

\begin{proof} Let $L^{1}$ and $L^{2}$ be a pair of mutually orthogonal Latin squares of order $d$. By Theorem 1, an absolutely maximally entangled basis can be constructed by using $L^{1}$. Assume that $|a_{0}+\alpha a_{1}+\cdots+\alpha^{d-1} a_{d-1}|=\sqrt{d}$, where $\alpha=1,\omega_{d},\ldots,\omega_{d}^{d-1}$.  We can obtain the following states from $L^{2}$,
$$
\left\{
\begin{array}{l}
|\psi_{n,m,0}\rangle=\frac{1}{d}\sum\limits_{k=0}^{d-1}\left(a_{k}\omega_{d}^{nk}|k\rangle \otimes \sum\limits_{i=0}^{d-1}\omega_{d}^{mi}|ij\rangle_k\right)\\ |\psi_{n,m,l}\rangle=\frac{1}{d}\sum\limits_{k=0}^{d-1}\left(a_{k}\omega_{d}^{nk}|k\rangle \otimes \sum\limits_{i=0}^{d-1}\omega_{d}^{mi}|ij\rangle_{k\oplus_{d} l}\right), ~l=1,2,\ldots,d-1,
\end{array}
\right.
$$
where $0\leq n,m\leq d-1$, $\omega_{d}=e^{\frac{2\pi\sqrt{-1}}{d}}$, $|ij\rangle_k$ and $|ij\rangle_{k\oplus_{d} l}$ have the same meaning as in Eq. $(*)$ and Eq. $(**)$ in Theorem 1, respectively. Obviously, $\{{|\psi_{n,m,l}\rangle},0\leq n,m,l\leq d-1\}$ is also a set of absolutely maximally entangled basis.

Moreover,
$$\centering {\begin{array}{cll}
\hspace{0cm}|\langle\phi_{n^{\prime},m^{\prime},0}\big| \psi_{n,m,0}\rangle|&=\frac{1}{d^{2}}|\sum\limits_{k,k^{\prime}=0}^{d-1}a_{k}\omega_{d}^{nk-n^{\prime}k^{\prime}}\langle k^{\prime}|k\rangle\sum\limits_{i,i^{\prime}=0}^{d-1}\omega_{d}^{mi-m^{\prime}i^{\prime}}(|i^{\prime}j^{\prime}\rangle_{k'},|ij\rangle_{k})\big| \\
\hspace{-1cm}&=\frac{1}{d^{2}}|\sum\limits_{k=0}^{d-1}\omega_{d}^{(n-n^{\prime})k}\omega_{d}^{(m-m^{\prime})k}a_{k}|\\
\hspace{-1cm}&=\frac{1}{d^{2}}|\sum\limits_{k=0}^{d-1}\omega_{d}^{[(n-n^{\prime})+(m-m^{\prime})]k}a_{k}|\\
\hspace{-1cm}&\triangleq\frac{1}{d^{2}}|a_{0}+\alpha a_{1}+\cdots+\alpha^{d-1} a_{d-1}|\\
\hspace{-1cm}&=\frac{1}{d\sqrt{d}}.
\end{array}}$$
In fact, since $L_1$ and $L_2$ are orthogonal, for any $0\leq k\leq d-1$ the pair of $(k,k)$ appears exactly once.

For any fixed $ 0\leq n^{\prime}, m^{\prime}, n,m \leq d- 1$ and $1\leq l^{\prime},l \leq d-1$, we have
$$\centering {\begin{array}{cll}
\hspace{0cm}|\langle\phi_{n^{\prime},m^{\prime},l^{\prime}}\big| \psi_{n,m,l}\rangle|&=\frac{1}{d^{2}}|\sum\limits_{k,k^{\prime}=0}^{d-1}a_{k}
\omega_{d}^{nk-n^{\prime}k^{\prime}}\langle k^{\prime}|k\rangle\sum\limits_{i,i^{\prime}=0}^{d-1}\omega_{d}^
{mi-m^{\prime}i^{\prime}}(|i^{\prime}j^{\prime}\rangle_{k'\oplus_dl^{\prime}},|ij\rangle_{k\oplus_dl})\big| \\
\hspace{-1cm}&=\frac{1}{d^{2}}|\sum\limits_{k=0}^{d-1}a_{k}\omega_{d}^{(n-n^{\prime})k}\omega_{d}^{(m-m^{\prime})i}|\\
\hspace{-1cm}&=\frac{1}{d^{2}}|\sum\limits_{k=0}^{d-1}a_{k}\omega_{d}^{(n-n^{\prime})k}|\\
\hspace{-1cm}&=\frac{1}{d\sqrt{d}}.
\end{array}}$$
Since the two Latin square after modulo addition is still orthogonal, the pair of $(k,k)$ also appears exactly once for any $0\leq k\leq d-1$. Therefore, we have the conclusion that the two absolutely maximally entangled bases are mutually unbiased.
\end{proof}

{\bf Remark.} In fact, the coefficients matrix of equation $|a_{0}+\alpha a_{1}+\cdots+\alpha^{d-1} a_{d-1}|=\sqrt{d}$ is a Fourier matrix. Some concrete solutions have been found in $\mathbb{C}^{d}$ (see references$^{[27,28]}$).
\begin{example}\label{}
Based on a complete set of MOLS of order $3$ in Fig. $1$, a pair of MUAMEBs in $\mathbb{C}^{3} \otimes \mathbb{C}^{3} \otimes \mathbb{C}^{3}$ can be obtained by Theorem \ref{2}. Without loss of generality, for the equation $|a_{0}+\alpha a_{1}+\alpha^{2} a_{2}|=\sqrt{3}$, we pick a solution as $(a_{0},a_{1},a_{2})=(1,1,\omega_{3})$. Here we denote $|ij\rangle_k$ and $|ij\rangle_{k\oplus_{d} l}$ as $|ij\rangle$ for short.
\begin{equation}
\nonumber
\hspace{-0.6cm}
\left\{
\begin{array}{ll}
|\phi_{n,m,0}\rangle=\frac{1}{3}& [|0\rangle(|00\rangle+\omega_{3}^{m}|11\rangle+\omega_{3}^{2m}|22\rangle)
+\omega_{3}^{n}|1\rangle(|02\rangle+\omega_{3}^{m}|10\rangle\\[2mm]
&
+\omega_{3}^{2m}|21\rangle)
+\omega_{3}^{2n}|2\rangle(|01\rangle+\omega_{3}^{m}|12\rangle+\omega_{3}^{2m}|20\rangle)]\\[2mm]
|\phi_{n,m,1}\rangle=\frac{1}{3}& [|0\rangle(|02\rangle+\omega_{3}^{m}|10\rangle+\omega_{3}^{2m}|21\rangle) +\omega_{3}^{n}|1\rangle(|01\rangle+\omega_{3}^{m}|12\rangle\\[2mm]
&
+\omega_{3}^{2m}|20\rangle)+
\omega_{3}^{2n}|2\rangle(|00\rangle+\omega_{3}^{m}|11\rangle+\omega_{3}^{2m}|22\rangle)]\\[2mm]
|\phi_{n,m,2}\rangle=\frac{1}{3}& [|0\rangle(|01\rangle+\omega_{3}^{m}|12\rangle+\omega_{3}^{2m}|20\rangle) +\omega_{3}^{n}|1\rangle(|00\rangle+\omega_{3}^{m}|11\rangle\\[2mm]
&
+\omega_{3}^{2m}|22\rangle)+
\omega_{3}^{2n}|2\rangle(|02\rangle+\omega_{3}^{m}|10\rangle+\omega_{3}^{2m}|21\rangle)]
\end{array}
\right.
\end{equation}
\begin{equation}
\nonumber
\hspace{-0.6cm}
\left\{
\begin{array}{ll}
|\psi_{n,m,0}\rangle=\frac{1}{3}&
[|0\rangle(|00\rangle+\omega_{3}^{m}|12\rangle+\omega_{3}^{2m}|21\rangle)
+\omega_{3}^{n}|1\rangle(|01\rangle+\omega_{3}^{m}|10\rangle\\[2mm]
&
+\omega_{3}^{2m}|22\rangle)+\omega_{3}\omega_{3}^{2n}|2\rangle
(|02\rangle+\omega_{3}^{m}|11\rangle+\omega_{3}^{2m}|20\rangle)]\\[2mm]
|\psi_{n,m,1}\rangle=\frac{1}{3}&
[|0\rangle(|01\rangle+\omega_{3}^{m}|10\rangle+\omega_{3}^{2m}|22\rangle)
+\omega_{3}^{n}|1\rangle(|02\rangle+\omega_{3}^{m}|11\rangle\\[2mm]
&
+\omega_{3}^{2m}|20\rangle)+\omega_{3}\omega_{3}^{2n}|2\rangle
(|00\rangle+\omega_{3}^{m}|12\rangle+\omega_{3}^{2m}|21\rangle)]\\[2mm]
|\psi_{n,m,2}\rangle=\frac{1}{3}&
[|0\rangle(|02\rangle+\omega_{3}^{m}|11\rangle+\omega_{3}^{2m}|20\rangle)
+\omega_{3}^{n}|1\rangle(|00\rangle+\omega_{3}^{m}|12\rangle\\[2mm]
&
+\omega_{3}^{2m}|21\rangle)+\omega_{3}\omega_{3}^{2n}|2\rangle
(|01\rangle+\omega_{3}^{m}|10\rangle+\omega_{3}^{2m}|22\rangle)]
\end{array}
\right.
\end{equation}
where $0\leq n,m\leq2$. It is easy to prove that the above two absolutely maximally entangled bases are mutually unbiased.
\end{example}

\section{Construction of MUAMEBs in $\mathbb{C}^{d_{1}} \otimes \mathbb{C}^{d_{2}} \otimes \mathbb{C}^{d_{1}d_{2}}$ }
In this section, we construct MUAMEBs in tripartite quantum system $\mathbb{C}^{d_{1}} \otimes \mathbb{C}^{d_{2}} \otimes \mathbb{C}^{d_{1}d_{2}}$ $(2 \leq d_1 \leq d_2,$ $d_{1}d_{2}\neq 6)$ based on weak orthogonal Latin squares. In addition, we also construct a product basis, which is mutually unbiased with the given absolutely maximally entangled bases.

\begin{theorem}\label{}
For any given $t$ mutually weak orthogonal Latin squares of order $d_{1}d_{2}$, there exist $t$ MUAMEBs in $\mathbb{C}^{d_{1}} \otimes \mathbb{C}^{d_{2}} \otimes \mathbb{C}^{d_{1}d_{2}}$.
\end{theorem}

\begin{proof} Suppose that $\{L^{1},L^{2},\ldots,L^{t}\}$ is a set of $t$ mutually weak orthogonal Latin squares of order $d_{1}d_{2}$.
Let $f: Z_{d_{1}}\times Z_{d_{2}}\longrightarrow Z_{d_{1}d_{2}}$ be a bijection, where $(x,y)\longmapsto d_{2}x+y$.
Denote
$$
|\phi^{s}_{n,m,l}\rangle=\frac{1}{\sqrt{d_{1}d_{2}}}\sum\limits_{i=0}^{d_{1}-1}\omega_{d_{1}}^{ni}
|i\rangle\otimes\sum\limits_{j=0}^{d_{2}-1}\omega_{d_{2}}^{mj}|j\rangle|{L^{s}_{lf(i,j)}}\rangle,
$$
where $\omega_{d}$=$e^{\frac{2\pi\sqrt{-1}}{d}}$, $0\leq n\leq d_{1}-1$, $0\leq m\leq d_{2}-1$, $0\leq l\leq d_{1}d_{2}-1$, $1\leq s\leq t$,  ${L^{s}_{lf(i,j)}}$ represent the element that appears in the $l$-th row and the $f(i,j)$-th column of the $s$-th Latin square.

(i)~Absolutely maximal entanglement.\\
With respect to the first subsystem, we have
$$\begin{array}{clll}
\rho_{1}&=tr_{23}(|\phi^{s}_{n,m,l}\rangle\langle\phi^{s}_{n,m,l}|)\\[2mm]
&=\frac{1}{d_{1}d_{2}}tr_{23}\left(\sum\limits_{i,i'=0}^{d_{1}-1}\omega_{d_{1}}^{n(i-i')}|i\rangle\langle i'|\otimes\sum\limits_{j,j'=0}^{d_{2}-1}\omega_{d_{2}}^{m(j-j')}|j L^{s}_{lf(i,j)}\rangle\langle j'  L^{s}_{lf(i',j')}|\right )
\\[2mm]
&=\frac{1}{d_{1}d_{2}}\sum\limits_{i,i^{\prime}=0}^{d_{1}-1}\omega_{d_{1}}^{n(i-i^{\prime})}|i\rangle\langle i^{\prime}|tr(\sum\limits_{j,j^{\prime}=0}^{d_{2}-1}\omega_{d_{2}}^{m(j-j^{\prime})}|j {L^{s}_{lf(i,j)}}\rangle\langle j^{\prime} L^{s}_{lf(i^{\prime},j^{\prime})}|)\\[2mm]
&=\frac{1}{d_{1}d_{2}}\sum\limits_{i=0}^{d_{1}-1}|i\rangle\langle i|\cdot d_{2}=\frac{I}{d_{1}}.
\end{array}$$
Similarly, it is easy to show that $\rho_{2}=\frac{I}{d_{2}}$. As for the third subsystem, we have
$$
\begin{array}{clll}
\rho_{3}&=tr_{12}(|\phi^{s}_{n,m,l}\rangle\langle\phi^{s}_{n,m,l}|)\\[2mm]
&=\frac{1}{d_{1}d_{2}}tr_{12}\left(\sum\limits_{i,i'=0}^{d_{1}-1}\omega_{d_{1}}^{n(i-i')}|i\rangle\langle i'|\otimes\sum\limits_{j,j'=0}^{d_{2}-1}\omega_{d_{2}}^{m(j-j')}|j L^{s}_{lf(i,j)}\rangle\langle j'  L^{s}_{lf(i',j')}|\right )\\[4mm]
&=\frac{1}{d_{1}d_{2}}\sum\limits_{i=0}^{d_{1}-1}\sum\limits_{j=0}^{d_{2}-1}|{L^{s}_{lf(i,j)}}\rangle\langle L^{s}_{lf(i,j)}|\\[3mm]
&=\frac{I}{d_{1}d_{2}}.
\end{array}$$
Thus $\{|\phi^{s}_{n,m,l}\rangle,\,0\leq n\leq d_{1}-1,\,0\leq m\leq d_{2}-1,\,0\leq l\leq d_{1}d_{2}-1\}$ is a set of absolutely maximally entangled states.

(ii)~Orthogonality.
$$\begin{array}{clll}
\langle\phi^{s}_{n^{\prime},m^{\prime},l^{\prime}}|\phi^{s}_{n,m,l}\rangle&=\frac{1}{d_{1}d_{2}}\sum\limits_{i,i^{\prime}=0}^{d_{1}-1}\omega_{d_{1}}^{ni-n^{\prime}i^{\prime}}\langle i^{\prime}|i\rangle\sum\limits_{j,j^{\prime}=0}^{d_{2}-1}\omega_{d_{2}}^{mj-m^{\prime}j^{\prime}}\langle j^{\prime}|j\rangle\langle L^{s}_{l^{\prime}f(i^{\prime},j^{\prime})}|{L^{s}_{lf(i,j)}}\rangle\\[4mm]
\hspace{-2.2cm}&=\frac{1}{d_{1}d_{2}}\sum\limits_{i=0}^{d_{1}-1}\omega_{d_{1}}^{(n-n^{\prime})i}\sum\limits_{j=0}^{d_{2}-1}\omega_{d_{2}}^{(m-m^{\prime})j}\langle L^{s}_{l^{\prime}f(i,j)}|{L^{s}_{lf(i,j)}}\rangle\\[3mm]
\hspace{-2.2cm}&=\delta_{n,n^{\prime}}\delta_{m,m^{\prime}}\delta_{l,l^{\prime}}.
\end{array}$$

(iii)~Unbiasedness.
$$\begin{array}{clll}
|\langle\phi^{s^{\prime}}_{n^{\prime},m^{\prime},l^{\prime}}|\phi^{s}_{n,m,l}\rangle|
&=\frac{1}{d_{1}d_{2}}|\sum\limits_{i,i^{\prime}=0}^{d_{1}-1}\omega_{d_{1}}^{ni-n^{\prime}i^{\prime}}\langle i^{\prime}|i\rangle\sum\limits_{j,j^{\prime}=0}^{d_{2}-1}\omega_{d_{2}}^{mj-m^{\prime}j^{\prime}}\langle j^{\prime}|j\rangle\langle L^{s^{\prime}}_{l^{\prime}f(i^{\prime},j^{\prime})}|{L^{s}_{lf(i,j)}}\rangle|\\[3mm]
&=\frac{1}{d_{1}d_{2}}|\sum\limits_{i=0}^{d_{1}-1}\omega_{d_{1}}^{(n-n^{\prime})i}\sum\limits_{j=0}^{d_{2}-1}\omega_{d_{2}}^{(m-m^{\prime})j}\langle L^{s^{\prime}}_{l^{\prime}f(i,j)}|{L^{s}_{lf(i,j)}}\rangle|\\[3mm]
&=\frac{1}{d_{1}d_{2}}|\omega_{d_{1}}^{(n-n^{\prime})i}\omega_{d_{2}}^{(m-m^{\prime})j}|\\[3mm]
&=\frac{1}{d_{1}d_{2}}.
\end{array}$$
Therefore, we obtain $t$ MUAMEBs in $\mathbb{C}^{d_{1}} \otimes \mathbb{C}^{d_{2}} \otimes \mathbb{C}^{d_{1}d_{2}}$ from a set of $t$ MWOLS$(d_{1}d_{2})$.
\end{proof}

\begin{lemma}\label{}
If $d_{1}d_{2}$ is prime power, there exist $(d_{1}d_{2}-1)$ mutually weak orthogonal Latin squares of order $d_{1}d_{2}$.
\end{lemma}

\begin{proof} Set $q=d_{1}d_{2}$. Let $F_{q}=\{0,1,\alpha,\ldots,\alpha^{q-2}\}$ be a finite field of order $q$, and $F^{\ast}_{q}=\{1,\alpha,\ldots,\alpha^{q-2}\}$ be the multiplication group of $F_{q}$. There exist $q-1$ $MOLS(q)$ as follows:
$$\begin{array}{c}
L^{1}= (i+j),~ L^{2}=(i+\alpha j),\ldots,~L^{q-1}=(i+\alpha^{q-2}j).
\end{array}$$
For any $L^{s_{0}}=(i+\alpha^{s_{0}}j)$ and $L^{s_{1}}=(i+\alpha^{s_{1}}j)$, if
$i_{0}+\alpha^{s_{0}}j=i_{1}+\alpha^{s_{1}}j$, we have $j=(\alpha^{s_1}-\alpha^{s_0})^{-1}(i_{0}-i_{1})$.
The same element of the $i_{0}$-th row and the $i_{1}$-th row only occurs in the $j$-th column.
Thus the Latin squares $L^{s_{0}}$ and $L^{s_{1}}$ are mutually weak orthogonal. Hence the $q-1$ Latin squares of order $q$ are also mutually weak orthogonal.
\end{proof}

From Theorem 3 and above Lemma we have
\begin{corollary}\label{}
From $(d_{1}d_{2}-1)$ \text{MWOLS}$(d_{1}d_{2})$, we can construct $(d_{1}d_{2}-1)$ MUAMEBs in $\mathbb{C}^{d_{1}} \otimes \mathbb{C}^{d_{2}} \otimes \mathbb{C}^{d_{1}d_{2}}.$
\end{corollary}
\begin{example}\label{}
Based on the three MWOLS of order $4$ given in Fig. $2$,
$$\centering {\begin{tabular}{c c c c}
{\small \begin{tabular}{|c|c|c|c|}
\hline
0 & 1 & 2 & 3\\
\hline
3 & 2 & 1 & 0\\
\hline
1 & 0 & 3 & 2\\
\hline
2 & 3 & 0 & 1\\
\hline
\end{tabular}} & {\small \begin{tabular}{|c|c|c|c|}
\hline
0 & 3 & 1 & 2\\
\hline
3 & 0 & 2 & 1\\
\hline
1 & 2 & 0 & 3\\
\hline
2 & 1 & 3 & 0\\
\hline
\end{tabular}} & {\small \begin{tabular}{|c|c|c|c|}
\hline
0 & 2 & 3 & 1\\
\hline
3 & 1 & 0 & 2\\
\hline
1 & 3 & 2 & 0\\
\hline
2 & 0 & 1 & 3\\
\hline
\end{tabular}}
\end{tabular}}$$
$$\centering \text{Fig}. 2~Three~\text{MWOLS}~of~order~4$$
we obtain three MUAMEBs in $\mathbb{C}^{2} \otimes \mathbb{C}^{2} \otimes \mathbb{C}^{4}$,
\begin{equation}
\nonumber
\centering
\left\{
\begin{array}{clll}
|\phi^{1}_{n,m,0}\rangle&=\frac{1}{2}[|0\rangle(|00\rangle+\omega_{2}^{m}|11\rangle)+\omega_{2}^{n}|1\rangle(|02\rangle+\omega_{2}^{m}|13\rangle)]\\
|\phi^{1}_{n,m,1}\rangle&=\frac{1}{2}[|0\rangle(|03\rangle+\omega_{2}^{m}|12\rangle)+\omega_{2}^{n}|1\rangle(|01\rangle+\omega_{2}^{m}|10\rangle)]\\
|\phi^{1}_{n,m,2}\rangle&=\frac{1}{2}[|0\rangle(|01\rangle+\omega_{2}^{m}|10\rangle)+\omega_{2}^{n}|1\rangle(|03\rangle+\omega_{2}^{m}|12\rangle)]\\
|\phi^{1}_{n,m,3}\rangle&=\frac{1}{2}[|0\rangle(|02\rangle+\omega_{2}^{m}|13\rangle)+\omega_{2}^{n}|1\rangle(|00\rangle+\omega_{2}^{m}|11\rangle)]
\end{array}
\right.
\end{equation}
\begin{equation}
\nonumber
\centering
\left\{
\begin{array}{clll}
|\phi^{2}_{n,m,0}\rangle=\frac{1}{2}[|0\rangle(|00\rangle+\omega_{2}^{m}|13\rangle)+\omega_{2}^{n}|1\rangle(|01\rangle+\omega_{2}^{m}|12\rangle)]\\
|\phi^{2}_{n,m,1}\rangle=\frac{1}{2}[|0\rangle(|03\rangle+\omega_{2}^{m}|10\rangle)+\omega_{2}^{n}|1\rangle(|02\rangle+\omega_{2}^{m}|11\rangle)]\\
|\phi^{2}_{n,m,2}\rangle=\frac{1}{2}[|0\rangle(|01\rangle+\omega_{2}^{m}|12\rangle)+\omega_{2}^{n}|1\rangle(|00\rangle+\omega_{2}^{m}|13\rangle)]\\
|\phi^{2}_{n,m,3}\rangle=\frac{1}{2}[|0\rangle(|02\rangle+\omega_{2}^{m}|11\rangle)+\omega_{2}^{n}|1\rangle(|03\rangle+\omega_{2}^{m}|10\rangle)]\\
\end{array}
\right.
\end{equation}
\begin{equation}
\nonumber
\centering
\left\{
\begin{array}{clll}
|\phi^{3}_{n,m,0}\rangle=\frac{1}{2}[|0\rangle(|00\rangle+\omega_{2}^{m}|12\rangle)+\omega_{2}^{n}|1\rangle(|03\rangle+\omega_{2}^{m}|11\rangle)]\\
|\phi^{3}_{n,m,1}\rangle=\frac{1}{2}[|0\rangle(|03\rangle+\omega_{2}^{m}|11\rangle)+\omega_{2}^{n}|1\rangle(|00\rangle+\omega_{2}^{m}|12\rangle)]\\
|\phi^{3}_{n,m,2}\rangle=\frac{1}{2}[|0\rangle(|01\rangle+\omega_{2}^{m}|13\rangle)+\omega_{2}^{n}|1\rangle(|02\rangle+\omega_{2}^{m}|10\rangle)]\\
|\phi^{3}_{n,m,3}\rangle=\frac{1}{2}[|0\rangle(|02\rangle+\omega_{2}^{m}|10\rangle)+\omega_{2}^{n}|1\rangle(|01\rangle+\omega_{2}^{m}|13\rangle)]\\
\end{array}
\right.
\end{equation}
where $0\leq n,m\leq1$.
\end{example}

\begin{theorem}\label{}
For any given pair of mutually weak orthogonal Latin squares of order $d_{1}$ and a pair of mutually weak orthogonal Latin squares of order $d_{2}$, there exists a pair of mutually weak orthogonal Latin squares of order $d_{1}d_{2}$. Furthermore, we can derive two MUAMEBs in $\mathbb{C}^{d_{1}} \otimes \mathbb{C}^{d_{2}} \otimes \mathbb{C}^{d_{1}d_{2}}$.
\end{theorem}

\begin{proof} Suppose that $L^{1}=(l^{1}_{i,j})_{d_{1}\times d_{1}}$ and $L^{2}=(l^{2}_{i,j})_{d_{1}\times d_{1}}$ be a pair of MWOLS$(d_{1})$;
$K^{1}=(k^{1}_{i,j})_{d_{2}\times d_{2}}$ and $K^{2}=(k^{2}_{i,j})_{d_{2}\times d_{2}}$ be a pair of MWOLS$(d_{2})$,
\begin{equation}
\nonumber
      L^{1}=
         \begin{array}{|c|c|c|}
         \hline
          l^{1}_{0,0} & \cdots & l^{1}_{0,d_{1}-1}\\
          \hline
          \vdots & \vdots & \vdots\\
          \hline
          l^{1}_{d_{1}-1,0} & \cdots & l^{1}_{d_{1}-1,d_{1}-1}\\
          \hline
         \end{array},~~~
        L^{2}=
         \begin{array}{|c|c|c|}
         \hline
         l^{2}_{0,0} & \cdots & l^{2}_{0,d_{1}-1}\\
         \hline
         \vdots & \vdots & \vdots\\
         \hline
         l^{2}_{d_{1}-1,0} & \cdots & l^{2}_{d_{1}-1,d_{1}-1}\\
         \hline
         \end{array},
\end{equation}
\begin{equation}
\nonumber
      K^{1}=
         \begin{array}{|c|c|c|}
         \hline
         k^{1}_{0,0} & \cdots & k^{1}_{0,d_{2}-1}\\
         \hline
         \vdots & \vdots & \vdots\\
         \hline
         k^{1}_{d_{2}-1,0} & \cdots & k^{1}_{d_{2}-1,d_{2}-1}\\
         \hline
         \end{array},~~~
      K^{2}=
         \begin{array}{|c|c|c|}
         \hline
         k^{2}_{0,0} & \cdots & k^{2}_{0,d_{2}-1}\\
         \hline
         \vdots & \vdots & \vdots\\
         \hline
         k^{2}_{d_{2}-1,0} & \cdots & k^{2}_{d_{2}-1,d_{2}-1}\\
         \hline
         \end{array}.
\end{equation}
We can construct two squares as follows:
$M^{1}=L^{1}\times K^{1}=(l^{1}_{i,j},k^{1}_{s,t})_{(i,s),(j,t)}$, $M^{2}=L^{2}\times K^{2}=(l^{2}_{i,j},k^{2}_{s,t})_{(i,s),(j,t)}$:
\begin{equation}
\nonumber
      M^{1}=
         \begin{array}{|c|c|c|}
         \hline
         l^{1}_{0,0}\times K^{1}& \cdots & l^{1}_{0,d_{1}-1}\times K^{1}\\
         \hline
         \vdots & \vdots & \vdots\\
         \hline
         l^{1}_{d_{1}-1,0}\times K^{1} & \cdots & l^{1}_{d_{1}-1,d_{1}-1}\times K^{1}\\
         \hline
         \end{array},\,
      M^{2}=
         \begin{array}{|c|c|c|}
         \hline
         l^{2}_{0,0}\times K^{2} & \cdots & l^{2}_{0,d_{1}-1}\times K^{2}\\
         \hline
         \vdots & \vdots & \vdots\\
         \hline
         l^{2}_{d_{1}-1,0}\times K^{2} & \cdots & l^{2}_{d_{1}-1,d_{1}-1}\times K^{2}\\
         \hline
         \end{array},
\end{equation}
where $l^{m}_{i,j}\times K^{m}=l^{m}_{i,j}\times k^{m}_{s,t}$, $i,j\in [d_{1}]$, $s,t\in [d_{2}]$.

It is easy to check that $M^{1}$ and $M^{2}$ is a pair of MWOLS$(d_{1}d_{2})$. In fact, for any two rows,
$${\begin{array}{c}
r^{1}_{(i,s)}=\{(l^{1}_{i,0},k^{1}_{s,0}),(l^{1}_{i,0},k^{1}_{s,1}),\ldots,(l^{1}_{i,0},k^{1}_{s,d_{2}-1}),\ldots,(l^{1}_{i,d_{1}-1},k^{1}_{s,d_{2}-1})\}\in M^{1},\\[1mm]
r^{2}_{(i^{\prime},s^{\prime})}=\{(l^{2}_{i^{\prime},0},k^{2}_{s^{\prime},0}),(l^{2}_{i^{\prime},0},k^{2}_{s^{\prime},1}),
\ldots,(l^{2}_{i^{\prime},0},k^{2}_{s^{\prime},d_{2}-1}),\ldots,(l^{2}_{i^{\prime},d_{1}-1},k^{2}_{s^{\prime},d_{2}-1})\}\in M^{2},
\end{array}}$$
since the Latin squares $L^{1}$ and $L^{2}$ are mutually weak orthogonal,
$(l^{1}_{i,0},l^{1}_{i,1},\ldots,l^{1}_{i,d_{1}-1})$ contains exactly one element in common with $(l^{2}_{i^{\prime},0},l^{2}_{i^{\prime},1},\ldots,l^{2}_{i^{\prime},d_{1}-1})$.
Similarly, $(k^{1}_{s,0},k^{1}_{s,1},\ldots,k^{1}_{s,d_{2}-1})$ also contains exactly one element in common with $(k^{2}_{s^{\prime},0},k^{2}_{s^{\prime},1},\ldots,k^{2}_{s^{\prime},d_{2}-1})$.
Without loss of generality, denote $l^{1}_{i,0}=l^{2}_{i^{\prime},0}$ and $k^{1}_{s,0}=k^{2}_{s^{\prime},0}$.
Then $(l^{1}_{i,0},k^{1}_{s,0})=(l^{2}_{i^{\prime},0},k^{2}_{s^{\prime},0})$ and
$r^{1}_{(i,s)}$ contains exactly one element in common with $r^{2}_{(i^{\prime},s^{\prime})}$.

Assume that
$\widetilde{{M^{1}}}=\{d_{2}l^{1}_{i,j}+k^{1}_{s,t}:i,j\in[d_{1}],s,t\in[d_{2}]\}$ and
$\widetilde{{M^{2}}}=\{d_{2}l^{2}_{i,j}+k^{2}_{s,t}:i,j\in [d_{1}],s,t\in [d_{2}]\}$.
we obtain two mutually weak orthogonal Latin squares $\widetilde{{M^{1}}}$ and $\widetilde{{M^{2}}}$. And by Theorem 3, we have that these two absolutely maximally entangled bases in $\mathbb{C}^{d_{1}} \otimes \mathbb{C}^{d_{2}} \otimes \mathbb{C}^{d_{1}d_{2}}$ are mutually unbiased.
\end{proof}

\begin{corollary}\label{}
If $d_{1}d_{2}=2^{s}p_{1}^{r_{1}}p_{2}^{r_{2}}\ldots p_{t}^{r_{t}} $ is a standard factorization of $d_{1}d_{2}$ into distinct prime
powers, where $p_{i}\neq p_{j}\neq2$, $p_{i}^{r_{i}}\geq3$, $s\geq2$, there exists a set of $min\{2^{s}-1,p_{1}^{r_{1}}-1,\ldots,p_{t}^{r_{t}}-1\}$~\text{MWOLS}$(d_{1}d_{2})$, that is, $min\{2^{s}-1,p_{1}^{r_{1}}-1,\ldots,p_{t}^{r_{t}}-1\}$ MUAMEBs in $\mathbb{C}^{d_{1}} \otimes \mathbb{C}^{d_{2}} \otimes \mathbb{C}^{d_{1}d_{2}}$ can be obtained.
\end{corollary}

Notice that the above construction mainly depends on the existence of mutually weak orthogonal Latin squares. However,
two Latin squares are orthogonal but not necessarily weak orthogonal. As orthogonal Latin squares have
a resolution of transversals, next we introduce another method to construct weak orthogonal Latin squares.

\begin{lemma}\label{}
If a Latin square $L^{1}$ of order $d$ has a resolution of transversals, then a Latin square $L^{2}$ of
order $d$ can be constructed and the two Latin squares are mutually weak orthogonal.
\end{lemma}

\begin{proof}
Set $[d]$ serve as the index set for the rows and columns in $L^{1}$. Without loss of generality, denote $L^{1}=(l_{i,j})$, $i,j\in[d]$. The 0-th column in $L^{1}$ is a natural arrangement of $0,1,\ldots,d-1$, namely, $l_{i,0}=i$, $i\in[d]$.

If $l^{k}$ be the transversal, then we take $k$ as the first entry in the 0-th column, $l^{k}=(l_{i_{0},0},l_{i_{1},1},\ldots,l_{i_{d-1},d-1})$, where $l_{i_{0},0}=l_{k,0}$. We take the transversal $l^{k}$ as the $k$-th row of $L^{2}$.
It is clear that $L^{2}$ is also a Latin square of order $d$. In fact, since the $i$-th row of $L^{2}$ is the transversal in the $L^{1}$, the entries in this row run through the index set. Since the transversal is picked from different rows and columns, each entry occurs exactly once in each row and exactly once in each column. Then for any two rows
$r^{1}=(l_{i,0},l_{i,1},\ldots,l_{i,d-1})\in L^{1}$ and $r^{2}=(l_{i_{0},0},l_{i_{1},1},$ $\ldots,l_{i_{d-1},d-1})\in L^{2}$,
only when $i=i_{s}$, $s\in[d]$, $r^{1}$ contains exactly one element in common with $r^{2}$. Thus the Latin squares $L^{1}$ and $L^{2}$ are mutually weak orthogonal.
\end{proof}

{\bf Remark.} When $d\neq 2,6$, there is at least a pair of mutually weak orthogonal Latin squares.

\begin{theorem}\label{}
There exists a product basis in $\mathbb{C}^{d_{1}} \otimes \mathbb{C}^{d_{2}} \otimes \mathbb{C}^{d_{1}d_{2}}$, which is mutually unbiased with the $t$ absolutely maximally entangled bases given in Theorem 3.
\end{theorem}

\begin{proof}
It is easy to prove that the following matrix $M$ in Fig. 3 is orthogonal to Latin square of order $d$. The matrix $M$ and Latin square $L$ are also mutually weak orthogonal.
$$\centering {\begin{array}{c c c c}
{\small \begin{array}{|c|c|c|c|}
\hline
0 & 0 & \cdots & 0\\
\hline
1 & 1 & \cdots & 1\\
\hline
\vdots & \vdots & \vdots & \vdots\\
\hline
d-1 & d-1 & d-1 & d-1\\
\hline
\end{array}}
\end{array}}$$
$$\centering \text{Fig}. 3~Matrix~M~of~order~d$$
Based on the matrix $M$ of order $d_{1}d_{2}$, we can construct a product basis in $\mathbb{C}^{d_{1}} \otimes \mathbb{C}^{d_{2}} \otimes \mathbb{C}^{d_{1}d_{2}}$ as follows,
$$
|\psi_{n,m,l}\rangle=\frac{1}{\sqrt{d_{1}d_{2}}}
\sum\limits_{i=0}^{d_{1}-1}\omega_{d_{1}}^{ni}|i\rangle\otimes
\sum\limits_{j=0}^{d_{2}-1}\omega_{d_{2}}^{mj}|j\rangle\otimes|{L_{lf(i,j)}}\rangle,
$$
where $\omega_{d}=e^{\frac{2\pi\sqrt{-1}}{d}}$, $0\leq n\leq d_{1}-1$, $0\leq m\leq d_{2}-1$ and $0\leq l\leq d_{1}d_{2}-1$.
Since for any fixed $l$, $L_{lf(i,j)}$ is the same as $M$ for all $0\leq i\leq d_1$ and $0\leq j\leq d_2$. Thus $\{|\psi_{n,m,l}\rangle,0\leq n\leq d_{1}-1,0\leq m\leq d_{2}-1, 0\leq l\leq d_{1}d_{2}-1\}$ is a product basis.

The product basis is mutually unbiased to the $t$ absolutely maximally entangled bases given in Theorem 3.
$$\begin{array}{clll}
|\langle\phi^{s}_{n^{\prime},m^{\prime},l^{\prime}}|\psi_{n,m,l}\rangle|
&=\frac{1}{d_{1}d_{2}}|\sum\limits_{i,i^{\prime}=0}^{d_{1}-1}\omega_{d_{1}}^{ni-n^{\prime}i^{\prime}}\langle i^{\prime}|i\rangle\sum\limits_{j,j^{\prime}=0}^{d_{2}-1}\omega_{d_{2}}^{mj-m^{\prime}j^{\prime}}\langle j^{\prime}|j\rangle\langle L^{s}_{l^{\prime}f(i^{\prime},j^{\prime})}|{L_{lf(i,j)}}\rangle|\\[3mm]
&=\frac{1}{d_{1}d_{2}}|\sum\limits_{i=0}^{d_{1}-1}\omega_{d_{1}}^{(n-n^{\prime})i}\sum\limits_{j=0}^{d_{2}-1}\omega_{d_{2}}^{(m-m^{\prime})j}\langle L^{s}_{l^{\prime}f(i,j)}|{L_{lf(i,j)}}\rangle|\\[3mm]
&=\frac{1}{d_{1}d_{2}}|\omega_{d_{1}}^{(n-n^{\prime})i}\omega_{d_{2}}^{(m-m^{\prime})j}|\\[3mm]
&=\frac{1}{d_{1}d_{2}}.
\end{array}$$
\end{proof}
\begin{example}\label{}
Based on the matrix $M$ of order $4$ in Fig. 3. and  Theorem 5, we obtain a product basis in $\mathbb{C}^{2} \otimes \mathbb{C}^{2} \otimes \mathbb{C}^{4}$,
\begin{equation}
\nonumber
\centering
\left\{
\begin{array}{clll}
|\psi_{n,m,0}\rangle&=\frac{1}{2}[(|0\rangle+\omega_{2}^{n}|1\rangle)(|0\rangle+\omega_{2}^{m}|1\rangle)|0\rangle]\\
|\psi_{n,m,1}\rangle&=\frac{1}{2}[(|0\rangle+\omega_{2}^{n}|1\rangle)(|0\rangle+\omega_{2}^{m}|1\rangle)|1\rangle]\\
|\psi_{n,m,2}\rangle&=\frac{1}{2}[(|0\rangle+\omega_{2}^{n}|1\rangle)(|0\rangle+\omega_{2}^{m}|1\rangle)|2\rangle]\\
|\psi_{n,m,3}\rangle&=\frac{1}{2}[(|0\rangle+\omega_{2}^{n}|1\rangle)(|0\rangle+\omega_{2}^{m}|1\rangle)|3\rangle]
\end{array}
\right.
\end{equation}
where $0\leq n,m\leq1$, which is mutually unbiased with the three absolutely maximally entangled bases given in Example~2.
\end{example}

For further illustration of the Theorem 3 and Lemma 2, see Example 4 in Appendix.

\section{Conclusion}
Based on a pair of mutually orthogonal Latin squares (MOLS), we have constructed a pair of mutually unbiased absolutely maximally entangled bases (MUAMEBs) in $\mathbb{C}^{d} \otimes \mathbb{C}^{d} \otimes \mathbb{C}^{d}$$(d\neq2,6)$. Then based on mutually weak orthogonal Latin squares (MWOLS), we have constructed MUAMEBs in tripartite quantum system $\mathbb{C}^{d_{1}} \otimes \mathbb{C}^{d_{2}} \otimes \mathbb{C}^{d_{1}d_{2}}$ $(d_{2}\geq d_{1}\geq2,d_{1}d_{2}\neq 6)$. In particular, when $d_{1}d_{2}$ is a prime power, there exist $(d_{1}d_{2}-1)$ MUAMEBs. When $d_{1}d_{2}=2^{s}p_{1}^{r_{1}}p_{2}^{r_{2}}\ldots p_{t}^{r_{t}},$ where $p_{i}\neq p_{j}\neq2$, $p_{i}^{r_{i}}\geq3$, $s\geq2$, we have obtained $min\{2^{s}-1,p_{1}^{r_{1}}-1,\ldots,p_{t}^{r_{t}}-1\}$ MUAMEBs. Moreover, we have presented a pair of MUAMEBs for other cases of $d_{1}d_{2}$. We have also derived a product basis, which is mutually unbiased to the above MUAMEBs. Besides, we have put forward the  corresponding examples in $\mathbb{C}^{3} \otimes \mathbb{C}^{3} \otimes \mathbb{C}^{3}$, $\mathbb{C}^{2} \otimes \mathbb{C}^{2} \otimes \mathbb{C}^{4}$ and $\mathbb{C}^{2} \otimes \mathbb{C}^{5} \otimes \mathbb{C}^{10}$. Here we list some specific results in the table below.
$$\begin{array}{|c|c|c|} \hline
d_{1}d_{2}& ~~M(d_{1},d_{2},d_{1}d_{2})& ~~N(d_{1},d_{2},d_{1}d_{2})\\ \hline
4 & 3(Corollary 1) & 4\\ \hline
8 & 7(Corollary 1) & 8\\ \hline
9 & 8(Corollary 1) & 9\\ \hline
10 & 2(Lemma 2) & 3\\ \hline
12 & 2(Corollary 2, Lemma 2) & 3\\ \hline
14 & 2(Lemma 2) & 3\\ \hline
15 & 2(Corollary 2) & 3\\ \hline
16 & 15(Corollary 1) & 16\\ \hline
18 & 2(Lemma 2) & 3\\ \hline
20 & 3(Corollary 2) & 4\\ \hline
 \end{array}$$
where $M(d_{1},d_{2},d_{1}d_{2})$ and $N(d_{1},d_{2},d_{1}d_{2})$ are the maximal number of MUAMEBs and MUBs in $\mathbb{C}^{d_{1}} \otimes \mathbb{C}^{d_{2}} \otimes \mathbb{C}^{d_{1}d_{2}}$, respectively.

\section*{Acknowledgements}
This work is supported by Natural Science Foundation of Hebei Province (F2021205001), NSFC (Grant Nos. 11871019, 62272208, 12075159, 12171044), Beijing Natural Science Foundation (Z190005), Academy for Multidisciplinary Studies, Capital Normal University, the Academician Innovation Platform of Hainan Province, and Shenzhen Institute for Quantum Science and Engineering, Southern University of Science and Technology (SIQSE202001).

\appendix
\section{Appendix}
The following example is provided to facilitate further understanding of our direct product construction.

\begin{example}\label{}
We present three mutually unbiased bases in $\mathbb{C}^{2} \otimes \mathbb{C}^{5} \otimes \mathbb{C}^{10}$.
Consider a pair of Latin squares and the matrix $M$ of order 10 given in the figure below, which are mutually weak orthogonal.

${\begin{tabular}{ll}
{\small \begin{tabular}{|c|c|c|c|c|c|c|c|c|c|}
\hline
0 & 4 & 1 & 7 & 2 & 9 & 8 & 3 & 6 & 5\\
\hline
8 & 1 & 5 & 2 & 7 & 3 & 9 & 4 & 0 & 6\\
\hline
9 & 8 & 2 & 6 & 3 & 7 & 4 & 5 & 1 & 0\\
\hline
5 & 9 & 8 & 3 & 0 & 4 & 7 & 6 & 2 & 1\\
\hline
7 & 6 & 9 & 8 & 4 & 1 & 5 & 0 & 3 & 2\\
\hline
6 & 7 & 0 & 9 & 8 & 5 & 2 & 1 & 4 & 3\\
\hline
3 & 0 & 7 & 1 & 9 & 8 & 6 & 2 & 5 & 4\\
\hline
1 & 2 & 3 & 4 & 5 & 6 & 0 & 7 & 8 & 9\\
\hline
2 & 3 & 4 & 5 & 6 & 0 & 1 & 8 & 9 & 7\\
\hline
4 & 5 & 6 & 0 & 1 & 2 & 3 & 9 & 7 & 8\\
\hline
\end{tabular}} & {\small \begin{tabular}{|c|c|c|c|c|c|c|c|c|c|}
\hline
0 & 1 & 2 & 3 & 4 & 5 & 6 & 7 & 9 & 8\\
\hline
1 & 3 & 6 & 8 & 0 & 7 & 9 & 2 & 4 & 5\\
\hline
2 & 5 & 8 & 6 & 7 & 9 & 0 & 1 & 3 & 4\\
\hline
3 & 7 & 9 & 4 & 6 & 2 & 8 & 5 & 0 & 1\\
\hline
4 & 8 & 5 & 7 & 9 & 6 & 1 & 0 & 2 & 3\\
\hline
5 & 6 & 0 & 1 & 2 & 3 & 4 & 8 & 7 & 9\\
\hline
6 & 0 & 1 & 2 & 3 & 4 & 5 & 9 & 8 & 7\\
\hline
7 & 9 & 3 & 5 & 1 & 8 & 2 & 4 & 6 & 0\\
\hline
8 & 4 & 7 & 9 & 5 & 0 & 3 & 6 & 1 & 2\\
\hline
9 & 2 & 4 & 0 & 8 & 1 & 7 & 3 & 5 & 6\\
\hline
\end{tabular}}
\end{tabular}}$

\begin{small}
\setlength{\arraycolsep}{1.2pt}
\hspace{3cm}\begin{tabular}{c c c c c c c c c c}
{\small \begin{tabular}{|c|c|c|c|c|c|c|c|c|c|}
\hline
0 & 0 & 0 & 0 & 0 & 0 & 0 & 0 & 0 & 0\\
\hline
1 & 1 & 1 & 1 & 1 & 1 & 1 & 1 & 1 & 1\\
\hline
2 & 2 & 2 & 2 & 2 & 2 & 2 & 2 & 2 & 2\\
\hline
3 & 3 & 3 & 3 & 3 & 3 & 3 & 3 & 3 & 3\\
\hline
4 & 4 & 4 & 4 & 4 & 4 & 4 & 4 & 4 & 4\\
\hline
5 & 5 & 5 & 5 & 5 & 5 & 5 & 5 & 5 & 5\\
\hline
6 & 6 & 6 & 6 & 6 & 6 & 6 & 6 & 6 & 6\\
\hline
7 & 7 & 7 & 7 & 7 & 7 & 7 & 7 & 7 & 7\\
\hline
8 & 8 & 8 & 8 & 8 & 8 & 8 & 8 & 8 & 8\\
\hline
9 & 9 & 9 & 9 & 9 & 9 & 9 & 9 & 9 & 9\\
\hline
\end{tabular}}
\end{tabular}
\end{small}

According to the Theorem~3 and Theorem~5, we get a pair of absolutely maximally entangled bases and a product base as follows:
\begin{equation}
\nonumber
\centering
\left\{
\begin{array}{clllllllll}
|\phi^{1}_{n,m,0}\rangle=\frac{1}{\sqrt{10}}&
[|0\rangle(|00\rangle+\omega_{5}^{m}|14\rangle+\omega_{5}^{2m}|21\rangle+\omega_{5}^{3m}|37\rangle+\omega_{5}^{4m}|42\rangle)\\
&+\omega_{2}^{n}|1\rangle(|09\rangle+\omega_{5}^{m}|18\rangle+\omega_{5}^{2m}|23\rangle+\omega_{5}^{3m}|36\rangle+\omega_{5}^{4m}|45\rangle)]\\
|\phi^{1}_{n,m,1}\rangle=\frac{1}{\sqrt{10}}&[|0\rangle(|08\rangle+\omega_{5}^{m}|11\rangle+\omega_{5}^{2m}|25\rangle+\omega_{5}^{3m}|32\rangle+\omega_{5}^{4m}|47\rangle)\\
&+\omega_{2}^{n}|1\rangle(|03\rangle+\omega_{5}^{m}|19\rangle+\omega_{5}^{2m}|24\rangle+\omega_{5}^{3m}|30\rangle+\omega_{5}^{4m}|46\rangle)]\\
|\phi^{1}_{n,m,2}\rangle=\frac{1}{\sqrt{10}}&[|0\rangle(|09\rangle+\omega_{5}^{m}|18\rangle+\omega_{5}^{2m}|22\rangle+\omega_{5}^{3m}|36\rangle+\omega_{5}^{4m}|43\rangle)\\
&+\omega_{2}^{n}|1\rangle(|07\rangle+\omega_{5}^{m}|14\rangle+\omega_{5}^{2m}|25\rangle+\omega_{5}^{3m}|31\rangle+\omega_{5}^{4m}|40\rangle)]\\
|\phi^{1}_{n,m,3}\rangle=\frac{1}{\sqrt{10}}&[|0\rangle(|05\rangle+\omega_{5}^{m}|19\rangle+\omega_{5}^{2m}|28\rangle+\omega_{5}^{3m}|33\rangle+\omega_{5}^{4m}|40\rangle)\\
&+\omega_{2}^{n}|1\rangle(|04\rangle+\omega_{5}^{m}|17\rangle+\omega_{5}^{2m}|26\rangle+\omega_{5}^{3m}|32\rangle+\omega_{5}^{4m}|41\rangle)]\\
|\phi^{1}_{n,m,4}\rangle=\frac{1}{\sqrt{10}}&[|0\rangle(|07\rangle+\omega_{5}^{m}|16\rangle+\omega_{5}^{2m}|29\rangle+\omega_{5}^{3m}|38\rangle+\omega_{5}^{4m}|44\rangle)\\
&+\omega_{2}^{n}|1\rangle(|01\rangle+\omega_{5}^{m}|15\rangle+\omega_{5}^{2m}|20\rangle+\omega_{5}^{3m}|33\rangle+\omega_{5}^{4m}|42\rangle)]\\
|\phi^{1}_{n,m,5}\rangle=\frac{1}{\sqrt{10}}&[|0\rangle(|06\rangle+\omega_{5}^{m}|17\rangle+\omega_{5}^{2m}|20\rangle+\omega_{5}^{3m}|39\rangle+\omega_{5}^{4m}|48\rangle)\\
&+\omega_{2}^{n}|1\rangle(|05\rangle+\omega_{5}^{m}|12\rangle+\omega_{5}^{2m}|21\rangle+\omega_{5}^{3m}|34\rangle+\omega_{5}^{4m}|43\rangle)]\\
|\phi^{1}_{n,m,6}\rangle=\frac{1}{\sqrt{10}}&[|0\rangle(|03\rangle+\omega_{5}^{m}|10\rangle+\omega_{5}^{2m}|27\rangle+\omega_{5}^{3m}|31\rangle+\omega_{5}^{4m}|49\rangle)\\
&+\omega_{2}^{n}|1\rangle(|08\rangle+\omega_{5}^{m}|16\rangle+\omega_{5}^{2m}|22\rangle+\omega_{5}^{3m}|35\rangle+\omega_{5}^{4m}|44\rangle)]\\
|\phi^{1}_{n,m,7}\rangle=\frac{1}{\sqrt{10}}&[|0\rangle(|01\rangle+\omega_{5}^{m}|12\rangle+\omega_{5}^{2m}|23\rangle+\omega_{5}^{3m}|34\rangle+\omega_{5}^{4m}|45\rangle)\\
&+\omega_{2}^{n}|1\rangle(|06\rangle+\omega_{5}^{m}|10\rangle+\omega_{5}^{2m}|27\rangle+\omega_{5}^{3m}|38\rangle+\omega_{5}^{4m}|49\rangle)]\\
|\phi^{1}_{n,m,8}\rangle=\frac{1}{\sqrt{10}}&[|0\rangle(|02\rangle+\omega_{5}^{m}|13\rangle+\omega_{5}^{2m}|24\rangle+\omega_{5}^{3m}|35\rangle+\omega_{5}^{4m}|46\rangle)\\
&+\omega_{2}^{n}|1\rangle(|00\rangle+\omega_{5}^{m}|11\rangle+\omega_{5}^{2m}|28\rangle+\omega_{5}^{3m}|39\rangle+\omega_{5}^{4m}|47\rangle)]\\
|\phi^{1}_{n,m,9}\rangle=\frac{1}{\sqrt{10}}&[|0\rangle(|04\rangle+\omega_{5}^{m}|15\rangle+\omega_{5}^{2m}|26\rangle+\omega_{5}^{3m}|30\rangle+\omega_{5}^{4m}|41\rangle)\\
&+\omega_{2}^{n}|1\rangle(|02\rangle+\omega_{5}^{m}|13\rangle+\omega_{5}^{2m}|29\rangle+\omega_{5}^{3m}|37\rangle+\omega_{5}^{4m}|48\rangle)]
\end{array}
\right.
\end{equation}
\begin{equation}
\nonumber
\centering
\left\{
\begin{array}{clllllllll}
|\phi^{2}_{n,m,0}\rangle=\frac{1}{\sqrt{10}}&[|0\rangle(|00\rangle+\omega_{5}^{m}|11\rangle+\omega_{5}^{2m}|22\rangle+\omega_{5}^{3m}|33\rangle+\omega_{5}^{4m}|44\rangle)\\
&+\omega_{2}^{n}|1\rangle(|05\rangle+\omega_{5}^{m}|16\rangle+\omega_{5}^{2m}|27\rangle+\omega_{5}^{3m}|39\rangle+\omega_{5}^{4m}|48\rangle)]\\
|\phi^{2}_{n,m,1}\rangle=\frac{1}{\sqrt{10}}&[|0\rangle(|01\rangle+\omega_{5}^{m}|13\rangle+\omega_{5}^{2m}|26\rangle+\omega_{5}^{3m}|38\rangle+\omega_{5}^{4m}|40\rangle)\\
&+\omega_{2}^{n}|1\rangle(|07\rangle+\omega_{5}^{m}|19\rangle+\omega_{5}^{2m}|22\rangle+\omega_{5}^{3m}|34\rangle+\omega_{5}^{4m}|45\rangle)]\\
|\phi^{2}_{n,m,2}\rangle=\frac{1}{\sqrt{10}}&[|0\rangle(|02\rangle+\omega_{5}^{m}|15\rangle+\omega_{5}^{2m}|28\rangle+\omega_{5}^{3m}|36\rangle+\omega_{5}^{4m}|47\rangle)\\
&+\omega_{2}^{n}|1\rangle(|09\rangle+\omega_{5}^{m}|10\rangle+\omega_{5}^{2m}|21\rangle+\omega_{5}^{3m}|33\rangle+\omega_{5}^{4m}|44\rangle)]\\
|\phi^{2}_{n,m,3}\rangle=\frac{1}{\sqrt{10}}&[|0\rangle(|03\rangle+\omega_{5}^{m}|17\rangle+\omega_{5}^{2m}|29\rangle+\omega_{5}^{3m}|34\rangle+\omega_{5}^{4m}|46\rangle)\\
&+\omega_{2}^{n}|1\rangle(|02\rangle+\omega_{5}^{m}|18\rangle+\omega_{5}^{2m}|25\rangle+\omega_{5}^{3m}|30\rangle+\omega_{5}^{4m}|41\rangle)]\\
|\phi^{2}_{n,m,4}\rangle=\frac{1}{\sqrt{10}}&[|0\rangle(|04\rangle+\omega_{5}^{m}|18\rangle+\omega_{5}^{2m}|25\rangle+\omega_{5}^{3m}|37\rangle+\omega_{5}^{4m}|49\rangle)\\
&+\omega_{2}^{n}|1\rangle(|06\rangle+\omega_{5}^{m}|11\rangle+\omega_{5}^{2m}|20\rangle+\omega_{5}^{3m}|32\rangle+\omega_{5}^{4m}|43\rangle)]\\
|\phi^{2}_{n,m,5}\rangle=\frac{1}{\sqrt{10}}&[|0\rangle(|05\rangle+\omega_{5}^{m}|16\rangle+\omega_{5}^{2m}|20\rangle+\omega_{5}^{3m}|31\rangle+\omega_{5}^{4m}|42\rangle)\\
&+\omega_{2}^{n}|1\rangle(|03\rangle+\omega_{5}^{m}|14\rangle+\omega_{5}^{2m}|28\rangle+\omega_{5}^{3m}|37\rangle+\omega_{5}^{4m}|49\rangle)]\\
|\phi^{2}_{n,m,6}\rangle=\frac{1}{\sqrt{10}}&[|0\rangle(|06\rangle+\omega_{5}^{m}|10\rangle+\omega_{5}^{2m}|21\rangle+\omega_{5}^{3m}|32\rangle+\omega_{5}^{4m}|43\rangle)\\
&+\omega_{2}^{n}|1\rangle(|04\rangle+\omega_{5}^{m}|15\rangle+\omega_{5}^{2m}|29\rangle+\omega_{5}^{3m}|38\rangle+\omega_{5}^{4m}|47\rangle)]\\
|\phi^{2}_{n,m,7}\rangle=\frac{1}{\sqrt{10}}&[|0\rangle(|07\rangle+\omega_{5}^{m}|19\rangle+\omega_{5}^{2m}|23\rangle+\omega_{5}^{3m}|35\rangle+\omega_{5}^{4m}|41\rangle)\\
&+\omega_{2}^{n}|1\rangle(|08\rangle+\omega_{5}^{m}|12\rangle+\omega_{5}^{2m}|24\rangle+\omega_{5}^{3m}|36\rangle+\omega_{5}^{4m}|40\rangle)]\\
|\phi^{2}_{n,m,8}\rangle=\frac{1}{\sqrt{10}}&[|0\rangle(|08\rangle+\omega_{5}^{m}|14\rangle+\omega_{5}^{2m}|27\rangle+\omega_{5}^{3m}|39\rangle+\omega_{5}^{4m}|45\rangle)\\
&+\omega_{2}^{n}|1\rangle(|00\rangle+\omega_{5}^{m}|13\rangle+\omega_{5}^{2m}|26\rangle+\omega_{5}^{3m}|31\rangle+\omega_{5}^{4m}|42\rangle)]\\
|\phi^{2}_{n,m,9}\rangle=\frac{1}{\sqrt{10}}&[|0\rangle(|09\rangle+\omega_{5}^{m}|12\rangle+\omega_{5}^{2m}|24\rangle+\omega_{5}^{3m}|30\rangle+\omega_{5}^{4m}|48\rangle)\\
&+\omega_{2}^{n}|1\rangle(|01\rangle+\omega_{5}^{m}|17\rangle+\omega_{5}^{2m}|23\rangle+\omega_{5}^{3m}|35\rangle+\omega_{5}^{4m}|46\rangle)]
\end{array}
\right.
\end{equation}
\begin{equation}
\nonumber
\centering
\left\{
\begin{array}{clllllllll}
|\psi_{n,m,0}\rangle=\frac{1}{\sqrt{10}}[(|0\rangle+\omega_{2}^{n}|1\rangle)(|0\rangle+\omega_{5}^{m}|1\rangle+\omega_{5}^{2m}|2\rangle+\omega_{5}^{3m}|3\rangle+\omega_{5}^{4m}|4\rangle)|0\rangle]\\
|\psi_{n,m,1}\rangle=\frac{1}{\sqrt{10}}[(|0\rangle+\omega_{2}^{n}|1\rangle)(|0\rangle+\omega_{5}^{m}|1\rangle+\omega_{5}^{2m}|2\rangle+\omega_{5}^{3m}|3\rangle+\omega_{5}^{4m}|4\rangle)|1\rangle]\\
|\psi_{n,m,2}\rangle=\frac{1}{\sqrt{10}}[(|0\rangle+\omega_{2}^{n}|1\rangle)(|0\rangle+\omega_{5}^{m}|1\rangle+\omega_{5}^{2m}|2\rangle+\omega_{5}^{3m}|3\rangle+\omega_{5}^{4m}|4\rangle)|2\rangle]\\
|\psi_{n,m,3}\rangle=\frac{1}{\sqrt{10}}[(|0\rangle+\omega_{2}^{n}|1\rangle)(|0\rangle+\omega_{5}^{m}|1\rangle+\omega_{5}^{2m}|2\rangle+\omega_{5}^{3m}|3\rangle+\omega_{5}^{4m}|4\rangle)|3\rangle]\\
|\psi_{n,m,4}\rangle=\frac{1}{\sqrt{10}}[(|0\rangle+\omega_{2}^{n}|1\rangle)(|0\rangle+\omega_{5}^{m}|1\rangle+\omega_{5}^{2m}|2\rangle+\omega_{5}^{3m}|3\rangle+\omega_{5}^{4m}|4\rangle)|4\rangle]\\
|\psi_{n,m,5}\rangle=\frac{1}{\sqrt{10}}[(|0\rangle+\omega_{2}^{n}|1\rangle)(|0\rangle+\omega_{5}^{m}|1\rangle+\omega_{5}^{2m}|2\rangle+\omega_{5}^{3m}|3\rangle+\omega_{5}^{4m}|4\rangle)|5\rangle]\\
|\psi_{n,m,6}\rangle=\frac{1}{\sqrt{10}}[(|0\rangle+\omega_{2}^{n}|1\rangle)(|0\rangle+\omega_{5}^{m}|1\rangle+\omega_{5}^{2m}|2\rangle+\omega_{5}^{3m}|3\rangle+\omega_{5}^{4m}|4\rangle)|6\rangle]\\
|\psi_{n,m,7}\rangle=\frac{1}{\sqrt{10}}[(|0\rangle+\omega_{2}^{n}|1\rangle)(|0\rangle+\omega_{5}^{m}|1\rangle+\omega_{5}^{2m}|2\rangle+\omega_{5}^{3m}|3\rangle+\omega_{5}^{4m}|4\rangle)|7\rangle]\\
|\psi_{n,m,8}\rangle=\frac{1}{\sqrt{10}}[(|0\rangle+\omega_{2}^{n}|1\rangle)(|0\rangle+\omega_{5}^{m}|1\rangle+\omega_{5}^{2m}|2\rangle+\omega_{5}^{3m}|3\rangle+\omega_{5}^{4m}|4\rangle)|8\rangle]\\
|\psi_{n,m,9}\rangle=\frac{1}{\sqrt{10}}[(|0\rangle+\omega_{2}^{n}|1\rangle)(|0\rangle+\omega_{5}^{m}|1\rangle+\omega_{5}^{2m}|2\rangle+\omega_{5}^{3m}|3\rangle+\omega_{5}^{4m}|4\rangle)|9\rangle]\\
\end{array}
\right.
\end{equation}
where $0\leq n\leq1$, $0\leq m\leq4.$  All the above bases are mutually unbiased.
\end{example}

\section*{Data availability statement}
All data that support the findings of this study are included within the article.


\begin{thebibliography}{}

\bibitem{B99} C. H. Bennett, G. Brassard, C. Cr$\acute{e}$peau, R. Jozsa, A. Peres, W. K. Wootters.: Teleporting an unknown quantum state via dual classical and Einstein-Podolsky-Rosen channels. {\it Phys. Rev. Lett.}, {\it 70} (1993) 1895-1899.

\bibitem{B99} V. Edral.: The role of relative entropy in quantum information theory. {\it Rev. Mod. Phys.}, {\it 74} (2002) 197-234.

\bibitem{W00} M. B. Plenio, S. Virmani.: An introduction to entanglement measures. {\it Quant. Inf. Comput.}, {\it 7} (2007) 1-51.

\bibitem{G01} R. Horodecki, P. Horodecki, M. Horodecki, K. Horodecki.: Quantum entanglement. {\it Rev. Mod. Phys.}, {\it 81} (2009) 865-942.

\bibitem{D03} J. Modlawska, A. Grudka.: Nonmaximally entangled states can be better for multiple linear optical teleportation. {\it Phys. Rev. Lett.},  {\it 100} (2008) 110503-1-110503-4.

\bibitem{H03} C. H. Bennett, D. P. DiVincenzo.: Quantum information and computation. {\it Nature.}, {\it 404} (2000) 247-255.

\bibitem{C04} C. H. Bennett.: Quantum cryptography using any two nonorthogonal states. {\it Phy. Rev. Lett.}, {\it 68} (1992) 3121.

\bibitem{R04} H. K. Lo, M. Curty, B. Qi.: Measurement-device-independent quantum key distribution. {\it Phys. Rev. Lett.}, {\it 108} (2012) 130503.

\bibitem{F04} W. K. Wootters, B. D.  Fields.: Optimal state-detemination by mutually unbaised measurements.  {\it Ann. Phys.}, {\it 191} (1989) 363-381.

\bibitem{E91} A. M. Steinberg, R. B. Adamson.: Experimental quantum state estimation with mutually unbiased bases. {\it Phys. Rev. Lett.}, {\it 105(3)} (2010) 030406.

\bibitem{B97} N. J. Cerf, M. Bourennane, A. Karlsson, N. Gisin.: Security of quantum key distribution using d-level systems. {\it Phys. Rev. Lett.}, {\it 88(12)} (2001) 127902.

\bibitem{B98} I. C. Yu, F. L. Lin, C. Y. Huang.: Quantum secret sharing with multi-level mutually  unbiased bases. {\it Phys. Rev. A.}, {\it 78(1)} (2008) 124-124.

\bibitem{S97} B. -G. Englert, D. Kaszlikowski, L. C. Kwek, W. H. Chee.: Wave-particle duality in multi-path interferometers: general concepts and three-path interferometers. {\it Int. J. Quant. Inf.}, {\it 6} (2008) 129-157.

\bibitem{G04} A. Peres.: Quantum Theory: Concepts and Methods. {\it Kluwer. Academic. Publishers.}, {\it Dordrecht}, (1995).

\bibitem{G97} Y. H. Tao, H. Nan, J. Zhang, S. M. Fei.: Mutually unbiased maximally entangled bases in $\mathbb{C}^{d} \otimes \mathbb{C}^{kd}$. {\it Quantum Inf. Process.}, {\it 14} (2015) 2291-2300.

\bibitem{D09} J. Y. Liu, M. H. Yang, K. Q. Feng.: Mutually unbiased maximally entangled bases in $\mathbb{C}^{d} \otimes \mathbb{C}^{d}$. {\it Quantum Inf. Process.}, {\it 16} (2017) 159.

\bibitem{D10} D. M. Xu.: Construction of mutually unbiased maximally entangled bases through permutations of Hadamard matrices. {\it Quantum Inf. Process.}, {\it 16(3)} (2017) 65.

\bibitem{Y11} Y. Y. Song, G. Y. Zhang, L. S. Xu, Y. H. Tao.: Construction of mutually unbiased bases using mutually orthogonal Latin squares. {\it Int. J. Theor.Phys.}, {\it 59} (2020) 1777-1787.

\bibitem{B11} Y. J. Zhang, H. Zhao, N. H, Jing, S. M. Fei.: Unextendible maximally entangled bases and mutually unbiased bases in multipartite systems. {\it Int. J. Theor. Phys.},  {\it 56} (2017) 3425-3430.

\bibitem{Y13} A. P. Street, D. J. Street.: Combinatorics of experimental design. {\it Journal of the American Statistical Association.}, {\it 84(405)} (1987) 576.

\bibitem{Z14} N. Starr.: The science behind sudoku by Jean-Paul Delahaye. {\it College Mathematics Journal.}, (2007).

\bibitem{F09} J. Hall, A. Rao.: Mutually orthogonal Latin squares from the inner products of vectors in mutually unbiased bases. {\it J. Phys. A.}, {\it 43(13)} (2010) 12.

\bibitem{Y12} P. Facchi, G. Florio, G. Parisi, S. Pascazio.: Maximally multipartite entangled states. {\it Phys. Rev. A.}, {\it 77} (2008) 060304.

\bibitem{B12} C. J. Colbourn, J. H. Dinitz.: Handbook of combinatorial designs. The CRC handbook of combinatorial designs., (1996).

\bibitem{C13} A. D. Keedwell, J. D$\acute{e}$nes.: Latin squares and their applications., (1974).

\bibitem{Y13} B. Musto.: Constructing mutually unbiased bases from quantum Latin squares. {\it Quantum Physics and Logic.}, {\it 236} (2017) 108-126.

\bibitem{C14} B. C. Hiesmayr, D. McNulty, S. Baek, S. SinghaRoy, J. Bae, D. Chru$\acute{s}$ci$\acute{n}$ski.: Detecting entanglement can be more effective with inequivalent mutually unbiased bases. {\it New J. Phys.}, {\it 23} (2021) 093018.

\bibitem{C18} S. Brierley, S. Weigert, I. Bengtsson.: All mutually unbiased bases in dimensions two to five. {\it Quant. Inf. Comput.}, {\it 10} (2010) 0803-0820.

\bibitem{C19} S. Bandyopadhyay, P. O. Boykin, V. Roychowdhury,
F. Vatan.: A new proof for the existence of mutually unbiased bases.  {\it Algorithmica.}, {\it 34(4)} (2002) 512-528.
\end{thebibliography}
\end{document}